\documentclass[12pt]{article}

\usepackage{diagbox}
\usepackage{mathtools}
\usepackage{bbm}
\usepackage{latexsym}
\usepackage{epsfig}
\usepackage{amsmath,amsthm,amssymb,enumerate}

\usepackage[a-1b]{pdfx}
\usepackage{hyperref}

\usepackage{color}

\def\reals{\mathbb R}

\newcounter{rot}%\addtocounter{rot}{1}, \therot

\addtocounter{rot}{1}

\def\a{\alpha} \def\b{\beta} \def\d{\delta} 
\def\e{\varepsilon} \def\f{\phi}   
\def\G{\Gamma}  \def\k{\kappa}
     \def\l{\lambda}

\def\t{\tau} \def\om{\omega}  \def\U{\Upsilon}

\def\wT{\widehat{T}}

\def\bx{{\bf x}}
\def\bC{{\bf C}}

\newtheorem{theorem}{Theorem}
\newtheorem{lemma}[theorem]{Lemma}

\newtheorem{Remark}{Remark}

\def\cT{{\mathcal T}}
\newcommand{\brac}[1]{\left(#1\right)}

\newcommand{\bfrac}[2]{\left(\frac{#1}{#2}\right)}

\newcommand{\set}[1]{\left\{#1\right\}}

\def\U{\Upsilon}

\def\Pr{\mathbb{P}}

\newcommand{\ignore}[1]{}

\def\cC{{\mathcal C}}

\def\cT{{\mathcal T}}

\newcommand{\beq}[2]{\begin{equation}\label{#1}#2\end{equation}}
\newcommand{\mults}[1]{\begin{multline*}#1\end{multline*}}

\def\beps{{\boldsymbol \epsilon}}

\def\cM{\mathcal{M}}

\def\cS{\mathcal{S}}

\usepackage{tikz}
\usetikzlibrary{decorations.pathmorphing}
\usetikzlibrary{positioning}
\usetikzlibrary{arrows,automata}
\usetikzlibrary{shapes.misc}
\usetikzlibrary{backgrounds}
\usetikzlibrary{arrows,shapes}

\def\bl{\boldsymbol{\l}}

\author{Alan Frieze\thanks{Research supported in part by NSF grant DMS1661063}
\ and 
Tomasz Tkocz\thanks{Research supported in part by NSF grant DMS1955175}\\
Carnegie Mellon University\\Pittsburgh PA15213\\U.S.A.
}

\begin{document}

\date{}
\title{Probabilistic analysis of algorithms for cost constrained minimum weighted combinatorial objects}
\maketitle
\begin{abstract}
We consider cost constrained versions of the minimum spanning tree problem and the assignment problem. We assume edge weights are independent copies of a continuous random variable $Z$ that satisfies $F(x)=\Pr(Z\leq x)\approx x^\a$ as $x\to0$, where $\a\geq 1$. Also, there are $r=O(1)$ budget constraints with edge costs chosen from the same distribution. We use Lagrangean duality to construct polynomial time algorithms that produce asymptotically optimal solutions. For the spanning tree problem, we allow $r>1$, but for the assignment problem we can only analyse the case $r=1$.
\end{abstract}
\section{Introduction}
Let $X$ be a finite set and let the elements of $X$ be given independent random weights $w(e),e\in X$ and random costs $c_i(e),i=1,2,\ldots,r$ for $e\in X$. Suppose we are given cost budgets of $\bC=(C_i,i=1,2,\ldots,r)$ and we consider the following problem: let $\cS$ denote some collection of of subsets of $X$. For a function $f:X\to\reals$ and $S\subseteq X$ we let $f(S)=\sum_{e\in S}f(e)$. We consider the optimization problem:
\[
Opt(\cS,\bC):\text{Minimise }w(S)\text{ subject to }S\in \cS_r=\set{S\in\cS:c_i(S)\leq C_i,i=1,2,\ldots,r},
\]
and let 
\[
\text{$w^*=w^*(\cS,\bC)$ denote the minimum value in $Opt(\cS,\bC)$.}
\]
This all sounds pretty general, but here we will only consider $X=E(K_n)$ or $X=E(K_{n,n})$, the edge sets of the complete graph and bipartite graph respectively. We will either have $\cS$ be the set of spanning trees of $K_n$ or the set of perfect matchings of $K_n$ or $K_{n,n}$.

In previous papers \cite{FPST}, \cite{FT1}, \cite{FT2} we focussed on giving high probability asymptotic estimates of $w^*$ in the case of trees, matchings, shortest paths and Hamilton cycles. In this paper we concentrate of finding polynomial time algorithms that w.h.p. find feasible solutions to $Opt(\cS,\bC)$ with weight $(1+o(1))w^*$. We do this without knowing asymptotic estimates of $w^*$. Finding such remain as open questions.
\paragraph{Spanning Trees}
Problems like this have been the focus of much research in the worst-case. For example Goemans and Ravi \cite{GR} consider the spanning tree problem with a single cost constraint. They give a polynomial time algorithm that finds a $(1+\e,1)$ solution to this problem. Here an $(\a,\b)$ solution $T$ is one that satisfies $c_i(T)\leq \a C_i,\,i=1,2,\ldots,r$ and $w(T)\leq \b w^*$. When $\a=1$ and $\b=1+o(1)$ we say that we have an {\em asymptotically optimal} solution. In \cite{FT1}, Frieze and Tkocz consider the case where $r=1$ and costs and weights are independent copies of the uniform $[0,1]$ random variable $U$ and give a polynomial time algorithm that w.h.p. finds an asymptotically optimal solution. When $r>1$ Grandoni, Ravi, Singh and Zenklusen \cite{GRSZ} consider a maximization version and give a $(1+\e,1)$ algorithm that runs in time $n^{O(r^2/\e)}$ time on an $n$-vertex graph. 

We will assume that we are working with real numbers. This may seem unrealistic, but we can instead work on a discretised version where we keep a polynomial number of bits $m$. It is routine to modify the proof below and make it work when $m=Kn\log n$, for large $K$. 

Our random weights and costs will be distributed as a continuous random variable $Z$ where $F(x)=\Pr(Z\leq x)\approx x^\a$ as $x\to0$, where $\a\geq 1$. Here $A\approx B$ is an abbreviation for $A=(1+o(1))B$ as $n\to\infty$, assuming that $A=A(n),B=B(n)$.

In this paper we prove
\begin{theorem}\label{th1}
Suppose that the edges of the complete graph $K_n$ are all given independent copies of $Z$. Let $\cT$ denote the set of spanning trees of $K_n$. Suppose there are $r=O(1)$ cost constraints whose coefficients are also independent copies of $Z$ or $Z_d$. Suppose also that $n\geq C_i=\om n^{1-r/(\a(r+1))}\log  n,i=1,2,\ldots,r$ where $\om\to\infty$. Then there is a polynomial time algorithm that w.h.p. fnds an asymptotically optimal solution to $Opt(\cT,\bC)$. 
\end{theorem}
\begin{Remark}\label{rem0}
Using a result of Gupta, Lee and Li \cite{GLL}, we can replace $K_n$ by an arbitrary dense regular graph $G$ with minimum degree $\d=\tilde{\Omega}(n^{1/2})$ and edge connectivity $\k\geq \d$. For this case, we will require that $n\geq C_i=\om n\d^{-1/(\a(r+1))}\log n$.
\end{Remark}

\paragraph{Matchings} Berger, Bonifaci, Grandoni and Sch\"afer \cite{BBGS} consider the case where $\cS=\cM$, the set of matchings of a graph. They consider the maximization version and describe an $n^{O(1/\e)}$ time algorithm that provides a $(1,1+\e)$ solution for the case where $r=1$. In this paper we prove
\begin{theorem}\label{th2}
Suppose that the edges of the complete graph $K_n$ or the complete bipartite graph $K_{n,n}$ are given independent copies of $Z$. Let $\cM$ denote the set of perfect matchings of the complete bipartite graph $K_{n,n}$. Suppose there a single cost constraint whose coefficients are also independent copies of $Z$. And suppose that the RHS $C_1\gg n^{1/2}$.  Then there is a polynomial time algorithm that w.h.p. fnds an asymptotically optimal solution to $Opt(\cM,\bC)$.
\end{theorem}
\section{Trees}\label{trees}
We consider the dual problem $Dual(\cT)$:
\beq{dualT}{
\text{Maximise }\f(\bl)\text{ over }\bl=(\l_1,\ldots,\l_r)\geq 0,\text{ where }\f(\bl)=\min\set{w(T)+\sum_{i=1}^r\l_i(c_i(T)-C_i): \text{$T\in\cT$ }}.
}
We note that
\beq{dual1}{
\text{if $\bl\geq 0$ and $T$ is feasible for $Opt(\cT,\bC)$ then $\f(\bl)\leq w(T)$.}
}
We will show that w.h.p. 
\beq{show}{
\text{that if $\bl^*$ solves \eqref{dualT} and $T^*$ solves $Opt(\cT,\cC)$ then $\f(\bl^*)\approx w(T^*)$. }
}

We note that solving \eqref{dualT} is equivalent to solving the Linear Program $LP(\cT)$:
\beq{primalT}{
\text{Minimise }\sum_{e\in E_n}w(e)x(e)\text{ subject to }\quad\bx\in P_\cT, \sum_{e\in E_n}c_i(e)x(e)\leq C_i, i=1,\ldots,r,\quad x(e)\geq 0, e\in E_n,
}
where $E_n=\binom{[n]}{2}$ and $P_\cT$ is the convex hull of the incidence vectors of the set $\cT$.

We also note that \eqref{show} implies that the relative integrality gap for the integer program $Opt(\cT,\cC)$ is $(1+o(1))$ w.h.p.

Next let 
\[
w_{\bl}(T)=w(T)+\sum_{i=1}^r\l_ic_i(T)\quad\text{ for $T\in \cT$.}
\]
Let $\cT_{\bl}$ denote the set of trees that minimise $w_{\bl}$ and $O_\cT(\bl)=\set{\bx(T):T\in \cT_{\bl}}$ denote the set of incidence vectors of the trees in $\cT_{\bl}$.
\begin{lemma}\label{lem1}
$|\cT_{\bl}|\leq r+1$ with probability one.
\end{lemma}
\begin{proof}
First assume that we are using $U$. Suppose that $\cT_{\bl}=\set{T_1,T_2,\ldots,T_s}$ where $s>r+1$ . Let $C$ be the $s\times r$ matrix $(C_{i,j}=c_i(T_j))$. Let $C_k$ be the matrix consisting of the first $k$ rows of $C$ and let ${\bf w_k}$ be the column vector $(w(T_j),j=1,2,\ldots,k)$ and let ${\bf 1}_k$ be the all ones vector of dimension $k$. Then we have $\f(\bl){\bf 1}_r={\bf w}_r+C_r\bl$. Now $C_r$ is non-singular with probability one and so $\bl=C_r^{-1}(\f(\bl){\bf 1}_r-{\bf w}_r)$. Rows $r+1,r+2$ of the equation $\f(\bl){\bf 1}_{r+2}={\bf w}_{r+2}+C_{r+2}\bl$ give us two distinct expressions for $\f(\bl)$. By equating them we find a non-trivial algebraic expression involving $w(T_j),c_i(T_j), i=1,2,\ldots,r,j =1,2,\ldots,s$ and such an expression exists with probability zero. The expression implies an explicit value for $w_{r+2}$, given the other parameters.
\end{proof}
\begin{Remark}
If we only keep weights/costs to $m$ bit accuracy then this claim has to be modified to be with probability $1-2^{-\Omega(m)}$. We can afford to use the union bound over all possible choices of $r$ spanning trees.
\end{Remark}
The optimum solution to $Opt(\cT,\cC)$ lies in the face of $P_\cT$ generated by the incidence vectors of the trees in $\cT_{\bl^*}$. They generate a face because they are the vertex solutions to a linear program. If $F$ is a face of a polytope $P$ and $E$ is an edge of $F$ then $E$ is an edge of $P$. Now if $T_1,T_2$ give rise to adjacent vertices of the polytope $P_\cT$ then $E(T_2)=(E(T_1)\setminus e)\cup\set{f}$ for edges $e,f$. It then follows from Lemma \ref{lem1} that we have
\begin{lemma}\label{lem2}
If $T_1,T_2$ minimise $w_{\bl}$ then $|E(T_1)\setminus E(T_2)|\leq r$, with probability one.
\end{lemma}
Next let $w_{\max}$ denote the maximum weight of any edge in any of the trees in $\cT_{\bl}$ and let $c_{\max}$ denote the maximum of any of the costs of any of the edges of any of the trees in $\cT_{\bl}$. 
\begin{lemma}\label{lem3}
With probability one, there exists $j$ such that $w(T_j)\leq w^*+rw_{\max}$ and $c_i(T_j)\leq C_i+rc_{\max}$ for $j=1,2,\ldots,r$.
\end{lemma}
\begin{proof}
Let $\bl=\bl^*$ solve $Dual(\cT)$ and $\bx^*$ solve $LP(\cT)$. Then we have that  $\bx^*$ is a convex combination of $\set{\bx(T):T\in O_\cT(\bl)}$. It follows that there exist $T_0,T_1,\ldots,T_r\in \cT_{\bl}$ such that (i) $w(T_0)\leq w^*$ and (ii) $c_i(T_i)\leq C_i,i=1,2,\ldots,r$. (If $c_i(T_j)>C_i,j\geq 1$ then $\f(\bl^*+\beps)>\f(\bl^*)$ for a sufficiently small perturbation $\beps$. This contradicts the fact that $\bl^*$ maximises $f$.)
\end{proof}
It then follows from Lemma \ref{lem3} that
\beq{Tmin}{
w(T)\leq w^*\text{ and }c_i(T)\leq C_i+r,\,i=1,2,\ldots,r.
}
This almost solves our problem, except that $T$ is not guaranteed to be feasible. We shownext that a small adjustement to $T$ results in an asymptotically optimal feasible solution. 

If $c_i(T)>C_i/2$ then $T$ contains at least $\frac{C_i}{4-C_i/(n-1)}$ edges $e\in X_i=\set{e\in T:c_i(e)\geq C_i/4n}$. Delete $4rn/C_i$ edges of $X_i$ from $T$, for each $i$, to create a forest $F$ for which $w(F)<w^*$ and $c_i(F)\leq C_i-r,\,i=1,2,\ldots,r$. Now observe that if $\G$ is the subgraph of $K_n$ spanned by edges $e$ for which 
\[
w(e)\leq \psi=F^{-1}(n^{-1/(r+1)}\log^{1/r}n)\approx n^{-1/(\a(r+1))}\log^{1/(\a r)}n,\qquad c_i(e)\leq \psi,i=1,2,\ldots,r
\]
 then $\G$ is distributed as $G_{n,p}$ where $p\approx n^{-1}\log^{(r+1)/r}n$. Thus $\G$ is connected w.h.p. and so we can add $4r^2n/C_{\min},(\,C_{\min}=\min_{i=1}^r C_i)$ edges from $\G$ to $F$ to make a spanning tree $\wT$. (The claim that $G_{n,p}$ is connected follows from Erd\H{o}s and R\'enyi \cite{ER}.) We have
\[
w(\wT)\leq w^*+\frac{4r^2n\psi}{C_{\min}}\text{ and }c_i(\wT)\leq C_i-r+\frac{4r^2n\psi}{C_{\min}}<C_i,\,i=1,2,\ldots,r,
\]
as our assumption on the $C_i$ implies that $\frac{n\psi}{C_{\min}}\to 0$. 

Also, to find the $s$ trees, we need only find one tree $T$ and then consider all trees of the form $T+e-f$.

\begin{Remark}
The argument that leads to \eqref{Tmin} is valid for an arbitrary matroid. 
\end{Remark}

Now consider the claim in Remark \ref{rem0}. Theorem 1.1 of \cite{GLL} implies the following: let $G$ be an arbitrary $\d$-regular graph as in Remark \ref{rem0}. If $\d p-\log n\to\infty$ then w.h.p. $G_p$ is connected. Here $G_p$ is obtained from $G$ by independently deleting edges with probability $1-p$.

As just observed we can w.h.p. find a tree $T$ satisfying \eqref{Tmin} in polynomial time.  If $c_i(T)>C_i/2$ then $T$ contains at least $\frac{C_i}{4-C_i/(n-1)}$ edges $e\in X_i=\set{e\in T:c_i(e)\geq C_i/4n}$. Delete $4rn/C_i$ edges of $X_i$ from $T$, for each $i$, to create a forest $F$ for which $w(F)<w^*$ and $c_i(F)\leq C_i-r,\,i=1,2,\ldots,r$. Now observe that if $\G$ is the subgraph of $G$ spanned by edges $e$ for which $w(e)\leq \eta=F^{-1}(\d^{-1/(r+1)}\log^{1/r}n),c_i(e)\leq \eta,i=1,2,\ldots,r$ then $\G$ is distributed as $G_{p}$ where $p\approx \d^{-1}\log^{(r+1)/r}n$. Thus $\G$ is connected w.h.p. and so we can add $4rn/C_i$ edges from $\G$ to $F$ to make a spanning tree $\wT$. We have
\[
w(\wT)\leq w^*+\frac{4^2n\eta}{C_i}\text{ and }c_i(\wT)\leq C_i-r+\frac{4r^2n\eta}{C_i}<C_i,\,i=1,2,\ldots,r,
\]
as our assumption on the $C_i$ implies that $\frac{n\eta}{C_i}\to 0$. 

\section{Matchings}\label{matchings} 
We analyze the algorithm of \cite{BBGS}, but we avoid the enumeration that gives a running time of $n^{O(1/\e)}$. We will only consider bipartite matchings. Our analysis only uses alternating paths and avoids the use blossoms and so the non-bipartite case is almost identical to the bipartite case.

We let $w_{\l}(M)=w(M)+\l (c_{1}(M)-C_{1})$ for $M\in \cM$ and $\l\geq 0$. We consider the dual problem $Dual(\cM,C_1)$: 
\beq{dual2}{
\text{Maximise :\;} \f(\l),\l\geq0\text{ where }\f(\l)=\min\set{w_\l(M):M\in \cM}.
}
We note that solving \eqref{dual2} is equivalent to solving the Linear program $LP(\cM)$:
\mults{
\text{Minimise }\sum_{i,j}w(i,j)x(i,j)\text{ subject to }\\
\sum_{i=1}^nx(i,j)=1,\forall j\in[n]\text{ and }\sum_{j=1}^nx(i,j)=1,\forall i\in[n]\text{ and }\sum_{i,j}c_1(i,j)x_{i,j}\leq C_1.
}
$LP(\cM)$ is a relaxation of $Opt(\cM,C_1)$ and so we should assume that its optimal solution is not integral. This would mean that the constraint $c_1(\bx)\leq C_1$ is tight at the optimum. Here $\bx=(x(i,j),i,j\in[n])$. 

Let $\cM^*(\l)$ denote the members of $\cM$ that minimise $w_{\l}$.
\begin{lemma}\label{cl1}
$|\cM^*(\l)|\leq 2$, with probability one.
\end{lemma}
\begin{proof}
First assume that we are using $U$. Suppose that there are three distinct members of $\cM$ that minimise  $w_{\l}$. This implies that there are three distinct $M_i\in \cM,i=1,2,3$ such that $w(M_i)+\l c_{1}(M_i)=C$ where $C=\f(\l)+\l C_{1}$. But this implies, after eliminating $C,\l$ that $\frac{w(M_1)-w(M_2)}{c_{1}(M_2)-c_{1}(M_1)}=\frac{w(M_1)-w(M_3)}{c_{1}(M_3)-c_{1}(M_1)}$, an event of probability zero.

In the case where we use $U_d$, given $M_1,M_2,M_3$, we see from the previous sentence that this probability is $2^{-m}$. There are $(n!)^3$ choices for the perfect matchings and we can use the union bound.
\end{proof}
So, there exist $M_1,M_2\in \cM$ that satisfy $w(M_1)+\l c_{1}(M_1)=w(M_2)+\l c_{1}(M_2)=\f_k(\l)$. 

Let $\bx(M)$ denote the $n^2$-dimensional $\set{0,1}$ index vector of matching $M$ and let $P_\cM$ denote the convex hull of these incidence vectiors. The optimum solution to $Opt(\cM,C_1)$ lies in the  line segment of $P_\cM$ generated by the incidence vectors of the two matchings minimizing $w_{\l^*}$. So, if we know at least one of $M_1,M_2$ and we know the optimum solution to $LP(\cM)$ then we can construct the other matching. We can find one of $M_1,M_2$ if we know $\l^*$. We just have to solve the assignment problem with weghts $w_{\l^*}$. Because $\f(\l)$ is a concave function, we can find $\l^*$ to within accuracy $2^{-poly(n)}$ by solving $poly(n)$ assignment problems. Alternatively, we can read off $\l^*$ from the solution to the dual of $LP(\cM)$:
\[
\text{Maximise }-C_1\l+\sum_{i=1}^nu_i+\sum_{j=1}^nv_j\text{ subject to }\l\geq 0\text{ and }-c_1(i,j)\l+u_i+v_j\leq w(i,j),\forall i,j. 
\]
Or
\[
\text{Maximise }_{\l\geq 0}\brac{\text{maximum }\sum_{i=1}^nu_i+\sum_{j=1}^nv_j\text{ subject to }
u_i+v_j\leq w_\l(i,j)}.
\]
Assume then that we know $\l^*,M_1,M_2$. Now we cannot have $C_{1}< \min\set{c_{1}(M_1), c_{1}(M_2)}$ else $\f(\l^*+\e)>\f(\l^*)$ for sufficiently small $\e$. This follows from Lemma \ref{cl1}. Assume then that $c_1(M_1)<C_1<c_1(M_2)$. We have
\beq{weak}{
w(M_1)+\l^*(c_{1}(M_1)-C_{1})=w(M_2)+\l^*(c_{1}(M_2)-C_{1})\leq w^*,
}
where the inequality come from weak duality.

Let $C=\set{e_1,e_2,\ldots,e_k}=M_1\oplus M_2$. Let $a_i=\d(e_i)w_{\l^*}(e_i)$ where $\d(e_i)=1$ for $i\in M_2$ and -1 otherwise. Then 
\[
w_{\l^*}(M_1)-w_{\l^*}(M_2)= \sum_{i=1}^ka_i=0
\]
 and so there exists $\ell$ such that 
\beq{1}{
\sum_{j=1}^ta_{\ell+j}\leq 0\text{ for }t=0,1,\ldots,k-1.
}
This is the content of the gasoline lemma of Lov\'asz \cite{L}, Problem 3.21. 

For $t\geq 0$ let $X_t=M_1\cup \set{e_{\ell+j}:j\leq t,\,\ell+j\text{ is odd}}\setminus \set{e_{\ell+j}:j\leq t,\,\ell+j\text{ is even}}$. Let $\t=\max\set{t:c(X_t)\leq C_1}$. Then we must have $\ell+\t$ even, because the $a_{2i-1}$ are positive and the $a_{2i}$ are negative. Note that $M=X_\t\setminus\set{e_1}$ is a matching and $c(M)\leq c(X_t)\leq C_1$ and that $|M|=n-1$. Note that \eqref{1} implies that 
\[
w_{\l^*}(M)\leq w_{\l^*}(X_\t)\leq w_{\l^*}(M_1).
\]
Now if $M^*$ solves $Opt(\cM,C_1)$ then
\[
w(M_1)=w_{\l^*}(M_1)-\l^* C_1\leq w_{\l^*}(M^*)-\l^* C_1\leq w_{\l^*}(M^*)-\l^* c_1(M^*) =w(M^*)
\]
as the $M_i$ minimize $w_{\l^*}$. So,
\mults{
w(X_\t)=w_{\l^*}(X_\t)-\l^* c_1(X_\t)=w_{\l^*}(X_\t)-\l^* C_1+\l^*(C_1- c_1(X_\t)) \leq\\ w_{\l^*}(M_1)-\l^* C_1+\l^*(C_1- c_1(X_\t)) \leq w^*+\l^*(C_1- c_1(X_\t)),
}
where the final inequality is from \eqref{weak}.

Let $f=x_{\ell+\t+1}$. Then, the maximality of $\t$ implies that $c_1(f)>C_1-c_1(X_\t)\geq 0$. So,
\beq{up1}{
w(M)\leq w(X_\t)\leq w^*+\l^* c_1(f)\leq w^*+\l^*.
}
Furthermore, by construction,
\beq{upC1}{
c(M)\leq C_1.
}
At this point, we need to do two things. The first is to bound $\l^*$ and the second is to deal with the fact that $|M|=n-1$. The following lemma deals with $\l^*$ 
\begin{lemma}\label{lsize}
%If $C_1\gg n^{1/2}$ then t
There is a constant $D=D(\a)>0$ depending only on $\a$ such that 
%\[
%\l^*\leq \frac{Dn}{C_1^2},\quad \text{with probability }1-O(n^{-8}).
%\]
\[
\l^*\leq \frac{Dn^{2-1/\a}}{C_1^2},\quad \text{w.h.p.}
\]
\end{lemma}
To deal with the second we prove the following lemma:
\begin{lemma}\label{Mn-1}
Suppose that $M$ is a matching of size at most $n-1$. Then w.h.p. there is an augmenting path that creates a matching $M'$ with (i) $|M'|=|M|+1$, (ii) $w(M')\leq w(M)+3n^{-1/3}$ and (iii) $c_1(M')\leq c_1(M)+3n^{-1/3}$.
\end{lemma}
It follows from \eqref{up1}, \eqref{upC1} and these two lemmas that we can w.h.p. find a perfect matching $M'$ such that 
\[
w(M')\leq w(M^*)+\frac{Dn^{2-1/\a}}{C_1^2}+3n^{-1/3}\text{ and }c_1(M')\leq c_1(M)+3n^{-1/3}.
\]
Now if $c_1(M)+3n^{-1/3}\leq C_1$ then we are done. Otherwise, we remove an edge from $M'$ of cost at least $3/n^{1/3}$. Such edges exist as we have assumed that $C_1\gg n^{2/3}$. Applying Lemma \ref{Mn-1} again we have a perfect matching $M''$ satisfying
\[
w(M'')\leq w(M^*)+\frac{Dn^{2-1/\a}}{C_1^2}+6n^{-1/3}\text{ and }c_1(M'')\leq C_1.
\]
Note that \cite{FPST}, Theorem 3 implies that $w^*=\Omega(n^{1-1/\a})$ and this will complete the proof of Theorem \ref{th2}, since we have assumed that $C_1\gg n^{1/2}$.
\subsection{Proof of Lemma \ref{lsize}}
\begin{proof}

Here we assume that the weights and costs are i.i.d. copies of a continuous random variable $X$ with $\Pr\brac{X \leq t} \approx t^\a$ as $t \to 0$ so that the density of $X$, call it $f$, satisfies $f(x) \approx \a x^{\a-1}$ as $x \to 0$. For a fixed $\l > 0$, the density $f_Z$ of $Z = X + \l X'$, where $X'$ is an independent copy of $X$, satisfies
\[
f_Z(x) = \int_0^x \l^{-1}f(\l^{-1}t)f(x-t) dt \approx \l^{-\a}\a^2 \int_0^x [t(x-t)]^{\a-1} dt = \l^{-\a}D_\a x^{2\a-1}, \qquad D_\a = \a^2\frac{\Gamma(a+1)^2}{\Gamma(2a+2)},
\]
as $x \to 0$. Thus, $\Pr\brac{Z \leq t} \approx \frac{D_\a}{2\a}\l^{-\a}t^{2\a}$, so by the results from \cite{FPST} (see Theorem 3, the unconstrained case $r=0$ and Section 6) applied to a rescaled version of $Z$ (so that its CDF behaves like $t^{2\a}$), we have w.h.p.,
\[
\phi(\l) + C_1\l = \min_M w_\l(M) =  \l^{1/2}(2\a/D_\a)^{1/(2\a)}\Theta(n^{1-1/(2\a)})
\]

Let $\f_A(\l)=A\l^{1/2}n^{1-1/(2\a)}-C_1\l$ for constant $A>0$.
Thus w.h.p. for some constants $0<A<B$ depending only of $\a$, we have
\[
\f_A(\l)\leq \f(\l)\leq \f_B(\l).
\]
Moreover, $\f_X,X=A,B$ is maximised at $\l_X^*=\frac{X^2n^{2-1/\a}}{4C_1^2}$ and then 
\beq{fX}{
\f_X(\l)=\f_X(\l_X^*)-C_1\brac{\sqrt{\l}-\sqrt{\l_X^*}}^2\text{ where }\f_X(\l_X^*)=\frac{X^2n}{4C_1^2}.
}
If $\l^*$ maximises $\f$ then for $K=O(1)$ we have, using \eqref{fX},
\[
\f(\l^*)\geq \f(\l_A^*)=\frac{A^2n^{2-1/\a}}{4C_1}>\f_B(K^2\l_B^*)=(1-(K-1)^2)\frac{B^2n^{2-1/\a}}{4C_1}. 
\]
It follows that $\l^*\leq K^2\l_B$ where $(1-(K-1)^2)B^2=A^2/2$.
\end{proof}
\subsection{Proof of Lemma \ref{Mn-1}}
\begin{proof}
For this we consider the random bipartite graph $H$ which consists of those edges $e$ for which $w(e),c_1(e)\leq n^{-1/3}$. This is distributed as the random bipartite graph $G_{n,n,p}$ where $p=n^{-2/3}$. Suppose that $a\in B_1,b\in B_2$ are the vertices not covered by $M$. Next let $A$ be the set of vertices in $V_2$ that can be reached by an alternating path of length at most five. We first observe that the minimum/maximum degree in $H$ is at least $\approx np$ w.h.p. See for example Frieze and Karo\'nski \cite{FK}, Theorem 3.4. We show next that w.h.p.
\beq{expand}{
S\subseteq V_1,\,n^{1/4}\leq |S|\leq n_0=\frac{n^{2/3}}{\log n}\text{ implies }|N(S)|\geq np|S|/4.
}
Indeed, if $v\in V_2$ then $\Pr(v\in N(S))=1-(1-p)^{|S|}\geq p|S|/2$. So,
\[
\Pr(\neg\eqref{expand})\leq \sum_{s=n^{1/4}}^{n_0}\binom{n}{s}\Pr(Bin(n,ps/2)\leq nps/4)\leq \sum_{s=n^{1/4}}^{n_0}\bfrac{ne}{s}^s e^{-nps/16}=\sum_{s=n^{1/4}}^{n_0}\bfrac{ne^{1-np/16}}{s}^s=o(1).
\]
Given the property in \eqref{expand} we see that regardless of $M$, there are at least $np/2$ alternating paths of length two, ending in $V_1$. Then there must be at least $np/2\times np/4-2np\geq n^2p^2/10>n_0$ alternating paths of length four ending in $V_1$. Finally we see that $|A|\geq n_0np/4=\Omega(n/\log n)$ and similarly for $|B|$. We then observe that w.h.p. there is an edge of $H$ connecting $A$ and $B$. Indeed, the probability there is no such edge is at most $2^{2n}(1-p)^{\Omega(n^2/\log^2n)}=o(1)$. It follows that we can convert $M$ to a perfect matching at an additonal weight and cost of at most $3n^{-1/3}$.
\end{proof}


\begin{thebibliography}{99}
\bibitem{BBGS} A. Berger, V. Bonifac, F. Grandoni and G. Schaefer, Budgeted matching and budgeted matroid intersection via the gasoline puzzle, {\em Lecture Notes in Computer Science} 5035 (2007) 273-287.
%
\bibitem{ER} P. Erd\H{o}s and A. R\'enyi, On random graphs I, {\em Publ. Math. Debrecen} 6 (1959) 290-297.
%
\bibitem{FPST} A.M. Frieze, W. Pegden, G. Sorkin and T. Tkocz, \href{https://arxiv.org/pdf/1910.08977.pdf}{Minimum-weight combinatorial structures under random cost-constraints}.
%
\bibitem{FK}  A.M. Frieze and M. Karo\'nski, Introduction to Random Graphs, Cambridge University Press, 2015.
%
\bibitem{FT1} A.M. Frieze, and T. Tkocz, \href{https://arxiv.org/pdf/1905.01229.pdf}{A randomly weighted minimum spanning tree with a random cost constraint}.
%
\bibitem{FT2}  A.M. Frieze, and T. Tkocz, \href{https://arxiv.org/pdf/1907.03375.pdf}{A randomly weighted minimum arborescence with a random cost constraint}
%
\bibitem{GR} M. Goemans and R. Ravi, {\em The constrained minimum spanning tree problem}, Fifth Scandinavian Workshop on Algorithm Theory, LNCS 1097, Reykjavik, Iceland (1996) 66-75.
%
\bibitem{GLL} A. Gupta, E. Lee and J. Li, The Connectivity Threshold for Dense Graphs.
%
\bibitem{GRSZ} F. Grandoni, R. Ravi, M. Singh and R. Zenklusen, New approaches to multi-objective optimization, {\em Mathematical Programming} 146, (2014) 525-554.
%
\bibitem{L} L. Lov\'asz, Combinatorial Problems and Exercises, AMS Chelsea Publishing, 2nd Edition, 2007.
%
\bibitem{W} D. Walkup, On the expected value of a random asignment problem, {\em SIAM Journal on Computing} 8 (1979) 440-442.
\end{thebibliography}
\end{document}